\documentclass[12pt]{article}  
\usepackage{color,amssymb,amsthm,amsmath,amsxtra,fullpage,graphicx,hyperref,float,appendix,bm}

\usepackage{pdfpages}

\usepackage{fancyhdr}
\usepackage{xcolor}
\definecolor{dark-red}{rgb}{0.4,0.15,0.15}
\definecolor{dark-blue}{rgb}{0,0,0.45}
\hypersetup{
    colorlinks, linkcolor={dark-blue},
    citecolor={blue}, urlcolor={blue}}

\numberwithin{equation}{section}
\usepackage{pdflscape}

\newcounter{constnum}
\def\mod{\text{ mod }} 

\renewcommand{\qed}{\nobreak \ifvmode \relax \else
      \ifdim\lastskip<1.5em \hskip-\lastskip
      \hskip1.5em plus0em minus0.5em \fi \nobreak
      \vrule height0.75em width0.5em depth0.25em\fi}

\definecolor{gold}{rgb}{0.85,0.66,0.0}

\newtheorem{theorem}{Theorem}[section]

\theoremstyle{definition}
\newtheorem{example}[theorem]{Example}
\newtheorem{definition}[theorem]{Definition}

\title{Families of Multidimensional Arrays with Good Autocorrelation and 
Asymptotically Optimal Cross-correlation}
\author{Sam Blake}
\date{18 July 2019}

\begin{document}

\maketitle

\begin{abstract}
We introduce a construction for families of $2n$-dimensional arrays with 
asymptotically optimal pairwise cross-correlation. These arrays are constructed 
using a circulant array of $n$-dimensional Legendre arrays. We also introduce 
an application of these higher-dimensional arrays to high-capacity digital 
watermarking of images and video. 
\end{abstract}

\section{Background}

Biphase sequence families with low periodic off peak autocorrelation and low cross-correlation are 
highly sought after for CDMA wireless communications. This has been an active research area since 
the landmark paper of Gold in 1967\cite{Gold1967}, and numerous constructions of such families are 
known. \\

The concept of binary two-dimensional doubly periodic arrays with optimal off-peak 
autocorrelation was introduced by Gordon in 1966\cite{Gordon1966}. Such arrays are two-dimensional 
analogues of m-sequences. They are either solitary, or have very small family sizes called maximal 
connected sets\cite{Tirkel1994}, and higher-dimensional families are rare\cite{Green1985}. Perfect 
binary arrays in two and higher-dimensions have also been studied\cite{Jedwab1994}, but they are 
solitary, and most have unfavourable aspect ratios for applications. In 1988, L\"{u}ke surveyed 
existing constructions of two-dimensional arrays\cite{Luke1988}. In 1989, the Legendre sequences 
were generalised to two and higher-dimensional analogues\cite{Luke1989}\cite{Bomer1993}. \\

Two and three-dimensional arrays find applications in optics, where they are used for coded 
aperture imaging, or in structured light, where they are used for image alignment or 
registration. The first families of two-dimensional arrays were constructed in 1991 by Green 
et al\cite{Green1991}, where the small Kasami and No--Kumar sequences were interpreted as arrays. \\

In 1997, Tirkel et al, motivated by finding two-dimensional patterns for use as spread spectrum 
watermarks, constructed families of arrays. The arrays were of size $p \times p$, where $p$ is 
a prime number. Later this was extended to $p \times p-1$, $p \times p+1$, and 
$p-1 \times p+1$ \cite{Leukhin2013} and to higher dimensions\cite{Blake2014}. \\

The periodic cross-correlation of two $N$-dimensional arrays, \textbf{A} and
\textbf{B}, both of size $l_0 \times l_1 \times\cdots\times l_{N-1}$, for shift $s_0, s_1, \cdots, s_{N-1}$ is defined as 
$$\theta_{\textbf{A}, \textbf{B}}\left(s_0, s_1, \cdots,
  s_{N-1}\right) = \sum_{i_0 = 0}^{l_0 - 1}\sum_{i_1 = 0}^{l_1 - 1}
\cdots \sum_{i_{N-1} = 0}^{l_{N-1} - 1} A_{i_0,i_1, \cdots, i_{N-1}}
B_{i_0 + s_0, i_1+s_1, \cdots, i_{N-1} + s_{N-1}}^*.$$ Similarly, the
periodic autocorrelation of a $N$-dimensional array for shift $s_0,
s_1, \cdots, s_{N-1}$ is given by $\theta_{\textbf{A}}\left(s_0,
s_1, \cdots, s_{N-1}\right) = \theta_{\textbf{A},\textbf{A}}\left(s_0,
s_1, \cdots, s_{N-1}\right)$. $\theta_{\textbf{A}}(s_0,s_1, \cdots, s_{N-1})$ is called an
\textit{off-peak} autocorrelation if not all $s_i = 0 \mod l_i$. We denote the 
array of all autocorrelations and cross-correlations for all shifts as 
$\theta_{\textbf{A}}$ and $\theta_{\textbf{A},\textbf{B}}$ respectively. \\

Sequences with flat periodic autocorrelation which used properties of the Legendre symbol were first 
discovered by Lerner in 1958\cite{Lerner1958}, where the sequences were termed \textit{Legendre sequences}. 
In the same year, Zierler showed the Legendre sequences possessed \textit{flat} periodic 
autocorrelation\cite{Zierler1958}. \\

\begin{definition}
(Legendre Sequences)\cite{Zierler1958}
Let \textbf{s} be a sequence of length $p$, for $p$ an odd prime. Then 
\[
s_k = 
\begin{cases}
a & \text{if $k = 0$}\\
1 & \text{if $k$ is a quadratic residue mod $p$}\\
-1 & \text{otherwise}
\end{cases}
\]
for $0 \leq k < p$.
\end{definition}

\bigskip

If $a = 0$, then all off-peak autocorrelations of the Legendre sequences equal to -1. If $a = \pm 1$, 
then all off-peak autocorrelations of the Legendre sequences $-1$ when $p = 4k-1$, and $1$ and $-3$ 
otherwise. \\

\begin{example}
We construct the length 17 Legendre sequence and compute its periodic autocorrelations.

$$\textbf{s} = [0, 1, 1, -1, 1, -1, -1, -1, 1, 1, -1, -1, -1, 1, -1, 1, 1],$$

$$\theta_{\textbf{s}} = [16, -1, -1, -1, -1, -1, -1, -1, -1, -1, -1, -1, -1, -1, -1, -1, -1].$$
\end{example}

In 1990, B\"{o}mer and Antweiler introduced a construction of two-dimensional Legendre arrays of sizes $p \times p$ where 
$p$ is an odd prime\cite{Bomer1990}. These arrays possessed flat autocorrelation, with all off-peak autocorrelations 
equal to -1. \\

\begin{definition}
(Legendre arrays)\cite{Bomer1990} Let $\alpha$ be a primitive element in $\text{GF}\left(p^2\right)$, then every power of 
$\alpha$ can be expressed as $$\alpha^i = m\,\alpha + n$$ where $0 \leq i < p^2 - 1$ and $(m,n) \neq (0,0)$. Then the 
two-dimensional Legendre array, \textbf{A}, is given by 
\[
A_{m,n} = 
\begin{cases}
$a$  & \text{if $(m,n) = (0,0)$}\\
$+1$ & \text{if $\alpha^{2r} = m\, \alpha + n$}\\
$-1$ & \text{if $\alpha^{2r+1} = m\, \alpha + n$}
\end{cases}
\]
\end{definition}

\bigskip

If $a = 0$, then all the off-peak periodic autocorrelations are $-1$, and if $a = \pm 1$ then all the 
off-peak periodic autocorrelations are $\pm 1$ and $\pm 3$.

\begin{example}
Let $p=5$, then $x^2 + 4x + 2$ is a primitive polynomial in $\text{GF}\left(p^2\right)$, and the $5 \times 5$ 
Legendre array, \textbf{A}, is given by 
\[
\left[
\begin{array}{ccccc}
 0 & 1 & 1 & 1 & 1 \\
 -1 & -1 & -1 & 1 & 1 \\
 -1 & 1 & -1 & 1 & -1 \\
 -1 & -1 & 1 & -1 & 1 \\
 -1 & 1 & 1 & -1 & -1 \\
\end{array}
\right]
\]
where $a=0$. 
\end{example}

This construction readily generalises to $n$-dimensions by using a primitive polynomial of degree $n$ 
in $\text{GF}\left(p^n\right)$. 

\begin{example}
Consider the construction of a 4D Legendre array. Let $p=3$, then $x^4 + 2x^3 + 2$ is a primitive 
polynomial in $\text{GF}\left(p^4\right)$, and the $3 \times 3 \times 3 \times 3$ Legendre array, \textbf{A}, 
is given by 
{\small
$$\textbf{A} = \left[
\begin{array}{ccc}
 \left[
\begin{array}{ccc}
 0 & 1 & 1 \\
 -1 & -1 & 1 \\
 -1 & 1 & -1 \\
\end{array}
\right] & \left[
\begin{array}{ccc}
 1 & 1 & -1 \\
 1 & 1 & -1 \\
 -1 & 1 & 1 \\
\end{array}
\right] & \left[
\begin{array}{ccc}
 1 & -1 & 1 \\
 -1 & 1 & 1 \\
 1 & -1 & 1 \\
\end{array}
\right] \\
 \left[
\begin{array}{ccc}
 -1 & -1 & 1 \\
 -1 & -1 & -1 \\
 1 & 1 & -1 \\
\end{array}
\right] & \left[
\begin{array}{ccc}
 -1 & 1 & 1 \\
 -1 & -1 & 1 \\
 1 & 1 & 1 \\
\end{array}
\right] & \left[
\begin{array}{ccc}
 1 & -1 & 1 \\
 -1 & -1 & 1 \\
 -1 & -1 & -1 \\
\end{array}
\right] \\
 \left[
\begin{array}{ccc}
 -1 & 1 & -1 \\
 1 & -1 & 1 \\
 -1 & -1 & -1 \\
\end{array}
\right] & \left[
\begin{array}{ccc}
 1 & 1 & -1 \\
 -1 & -1 & -1 \\
 -1 & 1 & -1 \\
\end{array}
\right] & \left[
\begin{array}{ccc}
 -1 & 1 & 1 \\
 1 & 1 & 1 \\
 -1 & 1 & -1 \\
\end{array}
\right] \\
\end{array}
\right]$$}
\end{example}

In 2017, Blake and Tirkel introduced a construction for multi-dimensional, block-circulant perfect autocorrelation 
arrays\cite{Blake2012}\cite{Blake2017b}. A special case of this construction is a two-dimensional perfect array, 
constructed from a circulant array\cite[const. XII, pp. 38]{Blake2017a}. \\

\begin{definition}
\cite{Blake2017a} Let ${\bf a} = \left[a_0, a_1, \cdots, a_{n-1}\right]$ and $\textbf{c} =
  \left[c_0,c_1,\cdots, c_{n-1}\right]$ be perfect sequences -- each of length
  $n$. We construct an array, \textbf{S}, such that $$\textbf{S} =
  \left[S_{i,j}\right] = a_j \,c_{i+j \mod n},$$ where $0 \leq i,j < n$.
\end{definition}

\bigskip

The sequence, \textbf{a}, is termed the \textit{multiplication sequence}.

\section{The multidimensional construction}

We now introduce a construction for families of $2n$-dimensional arrays. \\

\begin{definition}
Let $A$ be a $n$-dimensional Legendre array of size $p \times p \times\cdots\times p$, 
where $p$ is an odd prime. Then we construct a family of $p$, $2n$-dimensional arrays, 
$\textbf{S}_m$, for $0 \leq m < p$, where 
\begin{equation*}
\textbf{S}_m = \left[S_{i_0,i_1, \cdots, i_{2n-1}}\right]_m = 
	A_{i_0,i_1,\cdots,i_{n-1}} \, A_{m\,i_0 + i_n \mod p, m\,i_1 + i_{n+1} \mod p,\cdots, m\,i_{n-1} + i_{2n-1} \mod p},
\end{equation*}
where $0\leq i_0,i_1,\cdots,i_{2n-1}<p$.
\end{definition}

\bigskip

Similarities between this construction and the 2D circulant construction are clearly evident. In particular, 
the multiplication sequence $a_j$ and its counterpart the multiplication array $A_{i_0,i_1,\cdots,i_{n-1}}$; 
and the circulant array of columns $\,c_{i+j \mod n}$ and its counterpart 
$A_{m\,i_0 + i_n \mod p, m\,i_1 + i_{n+1} \mod p,\cdots, m\,i_{n-1} + i_{2n-1} \mod p}$.

\begin{theorem}
When $a$ (the first entry in the multidimensional Legendre array) is zero, the magnitude of 
the off-peak periodic autocorrelation of $\textbf{S}_m$ is bounded by $p^n-1$.
\end{theorem}

\begin{proof}
The periodic autocorrelation of $\textbf{S}_m$ for an off-peak shift $s_0,s_1,\cdots,s_{2n-1}$ is given by 

\begin{align*}
&\theta_{\textbf{S}_m}\left(s_0,s_1,\cdots,s_{2n-1}\right)=\\
&\sum_{i_0=0}^{p-1}\cdots\sum_{i_{2n-1}=0}^{p-1}
		\left[S_{i_0, \cdots, i_{2n-1}}\right]_m
		\left[S_{i_0+s_0, \cdots, i_{2n-1}+s_{2n-1}}\right]_m^* \\
&= \sum_{i_0=0}^{p-1}\cdots\sum_{i_{2n-1}=0}^{p-1}
	A_{i_0,\cdots,i_{n-1}} \, A_{m\,i_0 + i_n \mod p,\cdots, m\,i_{n-1} + i_{2n-1} \mod p}\\
	&\qquad\times A_{i_0 + s_0,\cdots,i_{n-1} + s_{n-1}} \, A_{m\,i_0 + m\,s_0 + i_n + s_n \mod p,\cdots, m\,i_{n-1} + m\,s_{n-1} + i_{2n-1} + s_{2n-1} \mod p} \\
&=\sum_{i_0=0}^{p-1}\cdots\sum_{i_{n-1}=0}^{p-1}
	\bigg[ A_{i_0,\cdots,i_{n-1}} \, A_{i_0 + s_0,\cdots,i_{n-1} + s_{n-1}}
		\sum_{i_n=0}^{p-1}\cdots\sum_{i_{2n-1}=0}^{p-1}\\
		&\qquad A_{m\,i_0 + i_n \mod p,\cdots, m\,i_{n-1} + i_{2n-1} \mod p} \,
			A_{m\,i_0 + m\,s_0 + i_n + s_n \mod p,\cdots, m\,i_{n-1} + m\,s_{n-1} + i_{2n-1} + s_{2n-1} \mod p}\bigg].
\end{align*}
When at least one of $m\,s_0 + s_n, m\,s_1 + s_{n+1}, \cdots, m\,s_{n-1} + s_{2n-1} \neq 0 \mod p$, the 
innermost summation above is -1 (as \textbf{A} is a Legendre array), then 
\begin{align*}
\theta_{\textbf{S}_m}\left(s_0,s_1,\cdots,s_{2n-1}\right) &= 
	-\sum_{i_0=0}^{p-1}\cdots\sum_{i_{n-1}=0}^{p-1}
		A_{i_0,\cdots,i_{n-1}} \, A_{i_0 + s_0,\cdots,i_{n-1} + s_{n-1}} \\
&=
\begin{cases}
1 - p^n & s_0=s_1=\cdots=s_{n-1} = 0 \mod p \\
1 & \text{otherwise}
\end{cases}
\end{align*}
Otherwise, when $m\,s_0 + s_n = m\,s_1 + s_{n+1} = \cdots = m\,s_{n-1} + s_{2n-1} = 0 \mod p$, 
as $p$ is prime this implies $s_0 = s_1 =  \cdots = s_{n-1} \neq 0$, then 
\begin{equation*}
\theta_{\textbf{S}_m}\left(s_0,s_1,\cdots,s_{2n-1}\right) = 
	\left(p^n-1\right)\sum_{i_0=0}^{p-1}\cdots\sum_{i_{n-1}=0}^{p-1}
		A_{i_0,\cdots,i_{n-1}} \, A_{i_0 + s_0,\cdots,i_{n-1} + s_{n-1}} = 1 - p^n
\end{equation*}
and the magnitude of the bound on the autocorrelation is $p^n-1$.
\end{proof}

\begin{theorem}
When $a$ (the first entry in the multidimensional Legendre array) is zero, the magnitude 
of the periodic cross-correlation of any two distinct arrays 
$\textbf{S}_{m_1}$ and $\textbf{S}_{m_2}$ is bounded by $p^n+1$.
\end{theorem}

\begin{proof}
The periodic cross-correlation of two distinct arrays $\textbf{S}_{m_1}$ and $\textbf{S}_{m_2}$ 
for shift $s_0,s_1,\cdots,s_{2n-1}$ is given by
\begin{align*}
&\theta_{\textbf{S}_{m_1},\textbf{S}_{m_2}}\left(s_0,s_1,\cdots,s_{2n-1}\right)=\\
&\sum_{i_0=0}^{p-1}\cdots\sum_{i_{2n-1}=0}^{p-1}
		\left[S_{i_0, \cdots, i_{2n-1}}\right]_{m_1}
		\left[S_{i_0+s_0, \cdots, i_{2n-1}+s_{2n-1}}\right]_{m_2}^* \\
&= \sum_{i_0=0}^{p-1}\cdots\sum_{i_{2n-1}=0}^{p-1}
	A_{i_0,\cdots,i_{n-1}} \, A_{m_1\,i_0 + i_n \mod p,\cdots, m_1\,i_{n-1} + i_{2n-1} \mod p}\\
	&\qquad\times A_{i_0 + s_0,\cdots,i_{n-1} + s_{n-1}} \, A_{m_2\,i_0 + m_2\,s_0 + i_n + s_n \mod p,\cdots, 
		m_2\,i_{n-1} + m_2\,s_{n-1} + i_{2n-1} + s_{2n-1} \mod p} \\
&=\sum_{i_0=0}^{p-1}\cdots\sum_{i_{n-1}=0}^{p-1}
	\bigg[ A_{i_0,\cdots,i_{n-1}} \, A_{i_0 + s_0,\cdots,i_{n-1} + s_{n-1}}
		\sum_{i_n=0}^{p-1}\cdots\sum_{i_{2n-1}=0}^{p-1}\\
		&\qquad A_{m_1\,i_0 + i_n \mod p,\cdots, m_1\,i_{n-1} + i_{2n-1} \mod p} \,
			A_{m_2\,i_0 + m_2\,s_0 + i_n + s_n \mod p,\cdots, m_2\,i_{n-1} + m_2\,s_{n-1} + i_{2n-1} + s_{2n-1} \mod p}\bigg].
\end{align*}


When at least one of $m_2\,s_0 + s_n, m_2\,s_1 + s_{n+1}, \cdots, m_2\,s_{n-1} + s_{2n-1} \neq 0 \mod p$, the 
inner-most summation is the cross-correlation of two shifted Legendre arrays, which is $p^n-1$ at the 
peak, and $-1$ otherwise. Then the outer-most summation is the autocorrelation of a Legendre array, 
with $p^n-2$ terms multiplied by $-1$ and one term multiplied by $p^n-1$. The bound occurs when $p^n-1$ 
is muliplied by $1$, and the inbalance of the remaining terms in the correlations is $-2$ multiplied by 
$-1$, then the bound on  $\theta_{\textbf{S}_{m_1},\textbf{S}_{m_2}}\left(s_0,s_1,\cdots,s_{2n-1}\right)$ is $p^n+1$. \\

Otherwise, when $m_2\,s_0 + s_n, m_2\,s_1 + s_{n+1}, \cdots, m_2\,s_{n-1} + s_{2n-1} = 0 \mod p$, as 
$p$ is prime this implies $s_n,s_{n+1}, \cdots, s_{2n-1} \neq 0 \mod p$. Then the inner-most 
summation is (as before) the cross-correlation of two shifted Legendre arrays. The outer-most 
summation is the autocorrelation of a Legendre array. At the peak of both summations, we have 
$\theta_{\textbf{S}_{m_1},\textbf{S}_{m_2}}\left(s_0,s_1,\cdots,s_{2n-1}\right) = (p^n-1) - (p^n-2) = 1$, 
otherwise at an off-peak shift of the inner-most summation we have 
$\theta_{\textbf{S}_{m_1},\textbf{S}_{m_2}}\left(s_0,s_1,\cdots,s_{2n-1}\right) = -1 - (p^n-2) = 1-p^n$. \\

Therefore $\theta_{\textbf{S}_{m_1},\textbf{S}_{m_2}}\left(s_0,s_1,\cdots,s_{2n-1}\right)$ is bounded 
by $p^n+1$.
\end{proof}

These arrays are asymptotically optimal in the sense of the Welch bound\cite{Welch1991}\cite{Yu2006}. Each 
array has $p^{2n} - 2p^n + 1$ non-zero entries. Then the cross-correlation bound to peak ratio is given 
by $(p^n+1)/(p^{2n} - 2p^n + 1)$ and is asymptotic to the Welch bound of $p^n/p^{2n}$. For example, 
for $p=67$, the relative difference is $1.5\times10^{-5}\%$. \\

\begin{example}\label{ex:3}
We illustrate the construction with the smallest possible example. Let $n = 2$ and $p = 3$, then 
$x^2+2x+2$ is a primitive polynomial in $\text{GF}\left(3^2\right)$, and we construct the 
$3 \times 3 \times 3 \times 3$ arrays, $\textbf{S}_1$ and $\textbf{S}_2$ and compute their
autocorrelations and cross-correlations. 
{\small
$$\textbf{S}_1 = \left[
\begin{array}{ccc}
 \left[
\begin{array}{ccc}
 0 & 0 & 0 \\
 0 & 0 & 0 \\
 0 & 0 & 0 \\
\end{array}
\right] & \left[
\begin{array}{ccc}
 1 & 0 & 1 \\
 1 & -1 & -1 \\
 -1 & -1 & 1 \\
\end{array}
\right] & \left[
\begin{array}{ccc}
 1 & 1 & 0 \\
 -1 & 1 & -1 \\
 1 & -1 & -1 \\
\end{array}
\right] \\
 \left[
\begin{array}{ccc}
 1 & -1 & 1 \\
 0 & -1 & -1 \\
 1 & 1 & -1 \\
\end{array}
\right] & \left[
\begin{array}{ccc}
 1 & 1 & -1 \\
 -1 & 0 & -1 \\
 -1 & 1 & 1 \\
\end{array}
\right] & \left[
\begin{array}{ccc}
 1 & -1 & -1 \\
 1 & 1 & 0 \\
 -1 & 1 & -1 \\
\end{array}
\right] \\
 \left[
\begin{array}{ccc}
 1 & 1 & -1 \\
 1 & -1 & 1 \\
 0 & -1 & -1 \\
\end{array}
\right] & \left[
\begin{array}{ccc}
 1 & -1 & -1 \\
 -1 & -1 & 1 \\
 1 & 0 & 1 \\
\end{array}
\right] & \left[
\begin{array}{ccc}
 1 & -1 & 1 \\
 -1 & 1 & 1 \\
 -1 & -1 & 0 \\
\end{array}
\right] \\
\end{array}
\right]$$}

{\small
$$\theta_{\textbf{S}_1} = \left[
\begin{array}{ccc}
 \left[
\begin{array}{ccc}
 64 & -8 & -8 \\
 -8 & -8 & -8 \\
 -8 & -8 & -8 \\
\end{array}
\right] & \left[
\begin{array}{ccc}
 1 & -8 & 1 \\
 1 & 1 & 1 \\
 1 & 1 & 1 \\
\end{array}
\right] & \left[
\begin{array}{ccc}
 1 & 1 & -8 \\
 1 & 1 & 1 \\
 1 & 1 & 1 \\
\end{array}
\right] \\
 \left[
\begin{array}{ccc}
 1 & 1 & 1 \\
 -8 & 1 & 1 \\
 1 & 1 & 1 \\
\end{array}
\right] & \left[
\begin{array}{ccc}
 1 & 1 & 1 \\
 1 & -8 & 1 \\
 1 & 1 & 1 \\
\end{array}
\right] & \left[
\begin{array}{ccc}
 1 & 1 & 1 \\
 1 & 1 & -8 \\
 1 & 1 & 1 \\
\end{array}
\right] \\
 \left[
\begin{array}{ccc}
 1 & 1 & 1 \\
 1 & 1 & 1 \\
 -8 & 1 & 1 \\
\end{array}
\right] & \left[
\begin{array}{ccc}
 1 & 1 & 1 \\
 1 & 1 & 1 \\
 1 & -8 & 1 \\
\end{array}
\right] & \left[
\begin{array}{ccc}
 1 & 1 & 1 \\
 1 & 1 & 1 \\
 1 & 1 & -8 \\
\end{array}
\right] \\
\end{array}
\right]$$
}

{\small
$$\textbf{S}_2 = \left[
\begin{array}{ccc}
 \left[
\begin{array}{ccc}
 0 & 0 & 0 \\
 0 & 0 & 0 \\
 0 & 0 & 0 \\
\end{array}
\right] & \left[
\begin{array}{ccc}
 1 & 1 & 0 \\
 -1 & 1 & -1 \\
 1 & -1 & -1 \\
\end{array}
\right] & \left[
\begin{array}{ccc}
 1 & 0 & 1 \\
 1 & -1 & -1 \\
 -1 & -1 & 1 \\
\end{array}
\right] \\
 \left[
\begin{array}{ccc}
 1 & 1 & -1 \\
 1 & -1 & 1 \\
 0 & -1 & -1 \\
\end{array}
\right] & \left[
\begin{array}{ccc}
 1 & -1 & 1 \\
 -1 & 1 & 1 \\
 -1 & -1 & 0 \\
\end{array}
\right] & \left[
\begin{array}{ccc}
 1 & -1 & -1 \\
 -1 & -1 & 1 \\
 1 & 0 & 1 \\
\end{array}
\right] \\
 \left[
\begin{array}{ccc}
 1 & -1 & 1 \\
 0 & -1 & -1 \\
 1 & 1 & -1 \\
\end{array}
\right] & \left[
\begin{array}{ccc}
 1 & -1 & -1 \\
 1 & 1 & 0 \\
 -1 & 1 & -1 \\
\end{array}
\right] & \left[
\begin{array}{ccc}
 1 & 1 & -1 \\
 -1 & 0 & -1 \\
 -1 & 1 & 1 \\
\end{array}
\right] \\
\end{array}
\right]$$
}

{\small 
$$\theta_{\textbf{S}_2} = 
\left[
\begin{array}{ccc}
 \left[
\begin{array}{ccc}
 64 & -8 & -8 \\
 -8 & -8 & -8 \\
 -8 & -8 & -8 \\
\end{array}
\right] & \left[
\begin{array}{ccc}
 1 & 1 & -8 \\
 1 & 1 & 1 \\
 1 & 1 & 1 \\
\end{array}
\right] & \left[
\begin{array}{ccc}
 1 & -8 & 1 \\
 1 & 1 & 1 \\
 1 & 1 & 1 \\
\end{array}
\right] \\
 \left[
\begin{array}{ccc}
 1 & 1 & 1 \\
 1 & 1 & 1 \\
 -8 & 1 & 1 \\
\end{array}
\right] & \left[
\begin{array}{ccc}
 1 & 1 & 1 \\
 1 & 1 & 1 \\
 1 & 1 & -8 \\
\end{array}
\right] & \left[
\begin{array}{ccc}
 1 & 1 & 1 \\
 1 & 1 & 1 \\
 1 & -8 & 1 \\
\end{array}
\right] \\
 \left[
\begin{array}{ccc}
 1 & 1 & 1 \\
 -8 & 1 & 1 \\
 1 & 1 & 1 \\
\end{array}
\right] & \left[
\begin{array}{ccc}
 1 & 1 & 1 \\
 1 & 1 & -8 \\
 1 & 1 & 1 \\
\end{array}
\right] & \left[
\begin{array}{ccc}
 1 & 1 & 1 \\
 1 & -8 & 1 \\
 1 & 1 & 1 \\
\end{array}
\right] \\
\end{array}
\right]
$$}

{\small
$$\theta_{\textbf{S}_1,\textbf{S}_2} = \left[
\begin{array}{ccc}
 \left[
\begin{array}{ccc}
 -8 & 1 & 1 \\
 1 & 1 & 1 \\
 1 & 1 & 1 \\
\end{array}
\right] & \left[
\begin{array}{ccc}
 10 & 1 & 1 \\
 -8 & -8 & 10 \\
 -8 & 10 & -8 \\
\end{array}
\right] & \left[
\begin{array}{ccc}
 10 & 1 & 1 \\
 -8 & -8 & 10 \\
 -8 & 10 & -8 \\
\end{array}
\right] \\
 \left[
\begin{array}{ccc}
 10 & -8 & -8 \\
 1 & 10 & -8 \\
 1 & -8 & 10 \\
\end{array}
\right] & \left[
\begin{array}{ccc}
 10 & -8 & -8 \\
 10 & 1 & -8 \\
 10 & -8 & 1 \\
\end{array}
\right] & \left[
\begin{array}{ccc}
 10 & 10 & 10 \\
 -8 & -8 & 1 \\
 -8 & 1 & -8 \\
\end{array}
\right] \\
 \left[
\begin{array}{ccc}
 10 & -8 & -8 \\
 1 & 10 & -8 \\
 1 & -8 & 10 \\
\end{array}
\right] & \left[
\begin{array}{ccc}
 10 & 10 & 10 \\
 -8 & -8 & 1 \\
 -8 & 1 & -8 \\
\end{array}
\right] & \left[
\begin{array}{ccc}
 10 & -8 & -8 \\
 10 & 1 & -8 \\
 10 & -8 & 1 \\
\end{array}
\right] \\
\end{array}
\right]$$}
\end{example}

\begin{example}
In the following graphic we plot the $7 \times 7 \times 7 \times 7$ array, $\textbf{S}_1$. 


\begin{figure}[H]
\centering
\includegraphics[width=0.6\textwidth]{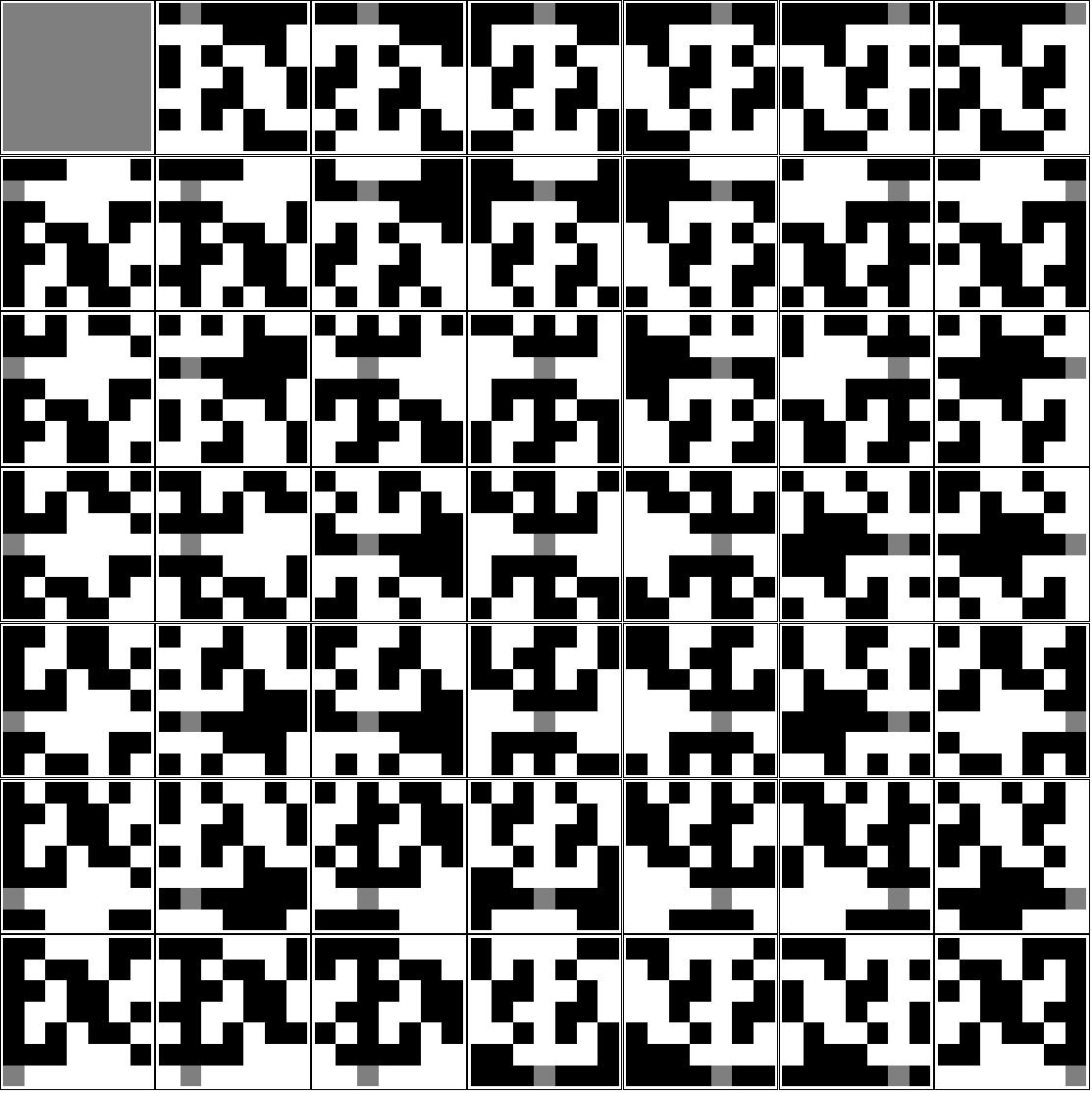}
\caption{A plot of $\textbf{S}_1$ for $p = 7$. ($-1$ is white, $+1$ is black and zero is gray.)}
\end{figure}
\end{example}

\section{Watermarking imagery with higher-dimensional arrays}

In this section we introduce a new application for higher-dimensional arrays with good 
autocorrelation and cross-correlation. We develop a new technique for embedding 
higher-dimensional arrays into imagery using spread spectrum watermarking 
techniques\cite{Tirkel1993}. \\

In the past, spread spectrum watermarking schemes used a sequence or array 
of dimensionality commensurate to the dimensionality of the dataset. Thus, a 
two-dimensional array is used to watermark an image; and every array embedded in 
the image supports a payload of two integers (the horizontal and vertical 
periodic shifts of the two-dimensional array). Then in order to support larger payloads, 
multiple arrays are embedded into the imagery. The embedding 
of multiple arrays relies on the cross-correlation properties of the families of 
arrays. However, even with families of optimal arrays, the signal-to-noise ratio (SNR) 
will decrease as more arrays are embedded. \\

We propose increasing the watermark payload, and subsequently increasing the SNR of the 
extraction, by increasing the dimensionality of the embedded arrays. We embed 
a $2n$-dimensional array into an image by partially flattening the array into a 
two-dimensional array. \\ 

We embed a $2n$-dimensional array, $\textbf{S} = \left[S_{i_0,i_1,\cdots,i_{2n-1}}\right]$ 
of size $d_0 \times d_1 \times\cdots\times d_{2n-1}$, 
into an image (2-dimensional dataset), $\textbf{I} = \left[I_{i,j}\right]$, by successively 
decreasing the dimensionality of $\textbf{S}$ from $2n$, $\textbf{S}^{2n}$ to $n$, 
$\textbf{S}^{n}$, until we have a 2-dimensional array, where
$$S^{n}_{i_0,i_1,\cdots,i_n} = S^{2n}_{q_0, q_1, \cdots, q_{n-1}, r_0, r_1,\cdots, r_{n-1}},$$
where $i_0 = q_0 d_0 + r_0, i_1 = q_1 d_1 + r_1, \cdots, i_{n-1} = q_{n-1} d_{n-1} + r_{n-1}$. \\

\begin{example}
We convert the $4$-D array, $\textbf{S}_1$, from Example \ref{ex:3} into a 2-D array. 
$$\left[
\begin{array}{ccccccccc}
 0 & 0 & 0 & 1 & 0 & 1 & 1 & 1 & 0 \\
 0 & 0 & 0 & 1 & -1 & -1 & -1 & 1 & -1 \\
 0 & 0 & 0 & -1 & -1 & 1 & 1 & -1 & -1 \\
 1 & -1 & 1 & 1 & 1 & -1 & 1 & -1 & -1 \\
 0 & -1 & -1 & -1 & 0 & -1 & 1 & 1 & 0 \\
 1 & 1 & -1 & -1 & 1 & 1 & -1 & 1 & -1 \\
 1 & 1 & -1 & 1 & -1 & -1 & 1 & -1 & 1 \\
 1 & -1 & 1 & -1 & -1 & 1 & -1 & 1 & 1 \\
 0 & -1 & -1 & 1 & 0 & 1 & -1 & -1 & 0 \\
\end{array}
\right]$$
\end{example}

Conversely, we partition the two-dimensional image into a 
$2n$-dimensional representation prior to cross-correlating with the family 
of $2n$-dimensional arrays. \\

This scheme embeds $2n$ integers (as periodic shifts of the array) for each array 
embedded, which greatly increases the payload per array in comparison with 
lower-dimensional arrays. 

\bibliographystyle{abbrv}

\begin{thebibliography}{99} 

\bibitem{Blake2012} S. Blake, T. E. Hall, A. Z. Tirkel, ``Arrays over Roots of Unity with Perfect Autocorrelation and Good ZCZ
  Cross--Correlation'', \textit{Advances in Mathematics of Communications (AMC)}, vol. 7, no. 3, pp. 231--242, 2013

\bibitem{Blake2014} S. Blake, O. Moreno, A.Z. Tirkel, ``Families of 3D Arrays for Video Watermarking'', SETA 2014, LNCS 8865, pp. 134--135, 2014

\bibitem{Blake2017a} S. Blake, ``Constructions for Perfect Autocorrelation Sequences and Multi-Dimensional Arrays'', 
\textit{PhD thesis}, Monash University, April 2017

\bibitem{Blake2017b} S. Blake, A. Z. Tirkel, ``A Multi-Dimensional Block-Circulant Perfect Array Construction", \textit{Advances in Mathematics of Communication}, accepted for publication 2017

\bibitem{Bomer1990} L. B\"{o}mer, M. Antweiler, ``Construction of a new class of higher-dimensional Legendre- and pseudonoise- arrays'', \textit{IEEE Symposium on IT 90}, pp. 76, San Diego, 1990

\bibitem{Bomer1993} L. B\"{o}mer, M. Antweiler, H. Schotten, ``Quadratic Residue Arrays'', \textit{Frequenz}, vol. 47, no. 7--8, pp. 190--196, 1993

\bibitem{Gold1967} R. Gold, ``Optimal binary sequences for spread spectrum multiplexing'', \textit{IEEE Trans. Inform. Theory},
vol. 13, issue 4, pp. 619--621, 1967

\bibitem{Gordon1966} B. Gordon, ``On the existence of perfect maps'', \textit{IEEE Trans. Inform. Theory}, vol. 12, issue 4, 
pp. 486--487, 1966

\bibitem{Green1985} D.H. Green, ``Structural Properties of Pseudo random Arrays and Volumes and their Related Sequences'', 
\textit{IEE Proceedings E-Computers and Digital Techniques}, vol. 132, issue 3, pp. 133 – 145, 1985

\bibitem{Green1991} D.H. Green, S.K. Amarasinghe, ``Families of Sequences and Arrays with Good Periodic Correlation'',
\textit{IEE Proc. E}, vol. 138, issue 4, pp. 260--268, 1991

\bibitem{Jedwab1994} J. Jedwab, C. Mitchell, F. Piper and P. Wild, ``Perfect binary arrays and difference sets'', 
\textit{Discr. Math.}, vol. 125, no. 1--3, pp. 241--254, 1994

\bibitem{Lerner1958} R. Lerner, ``Signals having uniform ambiguity functions'', \textit{IRE Convention Record, 
Information Theory Section}, March 1958

\bibitem{Leukhin2013} A. Leukhin, O. Moreno, A. Tirkel, ``Secure CDMA and Frequency Hop Sequences'', \textit{
The Tenth International Symposium on Wireless Communication Systems}, pp.819--825, Germany, 2013 

\bibitem{Luke1988} H. D. L\"{u}ke, ``Sequences and Arrays with Perfect Periodic Correlation'', 
	\textit{IEEE Transactions on Aerospace and Electronic Systems}, vol. 24, no. 3, pp. 287-294, 1988

\bibitem{Luke1989} H. D. L\"{u}ke, L. B\"{o}mer, M. Antweiler, ``Perfect binary arrays'', 
\textit{Signal Processing}, vol. 17, no. 1, pp. 69-80, 1989

\bibitem{Tirkel1993} A.Z. Tirkel, G.A. Rankin, R.M. Van Schyndel, W.J. Ho, N.R.A. Mee, 
	C. F. Osborne,  ``Electronic Water Mark'', DICTA 93, Macquarie University, pp. 666-673, 1993

\bibitem{Tirkel1994} A.Z. Tirkel, C.F. Osborne, N. Mee, G.A. Rankin, A. McAndrew, ``Maximal Connected Sets - Application to CDMA'', \textit{International Journal of Digital and Analog Communication Systems}, vol 7, pp. 29-32, 1994

\bibitem{Tirkel1997} A.Z. Tirkel, C.F. Osborne and T.E. Hall, ``Steganography - Applications of coding theory'', 
\textit{IEEE Information Theory Workshop}, pp. 57--59, Norway, 1997

\bibitem{Welch1991} L. R. Welch, ``Lower bounds on the maximum cross correlation of signals'', \textit{IEEE Trans. Inform. Theory}, vol. 20, no. 3, pp. 603--616, May 1991

\bibitem{Yu2006} N.Y. Yu, G. Gong, ``On asymptotic optimality of binary sequence families'', CACR2006-28, University of 
Waterloo, Canada, May 2006

\bibitem{Zierler1958} N. Zierler, ``Legendre Sequences'', MIT Lincoln Laboratory, Group Report 34-71, May 1958

\end{thebibliography}

\end{document}